\documentclass[10pt, journal, twocolumn]{IEEEtran}

\usepackage{longtable}
\usepackage{rotating}
\usepackage{multirow}
\usepackage{epsfig, amsmath,amssymb,epsf,cite,subfigure,scalefnt,multirow,array,setspace}
\usepackage{algorithm}
\usepackage{algorithmic}
\usepackage[top=2cm,bottom=2cm,left=2cm,right=2cm]{geometry}
\usepackage{times}
\usepackage{epstopdf}
\usepackage{subfig}
\usepackage{amsfonts}
\usepackage{booktabs}
\usepackage{multirow}
\usepackage{multicol}
\usepackage{cite}
\usepackage{algorithm}
\usepackage{algorithmic}
\usepackage{stfloats}
\usepackage{bm}
\usepackage{graphicx}
\usepackage{color}
\usepackage{enumerate}

\usepackage[final]{pdfpages}
\usepackage{fancyhdr}
\usepackage{lineno}

  \def\cC{{\mathcal{C}}}

 \def\cN{{\mathcal{N}}}

\def\ba{{\mathbf{a}}} \def\bb{{\mathbf{b}}}  \def\bd{{\mathbf{d}}} \def\be{{\mathbf{e}}}
\def\bff{{\mathbf{f}}}    
   \def\bn{{\mathbf{n}}} 
   \def\bs{{\mathbf{s}}} 
 \def\bv{{\mathbf{v}}}   \def\by{{\mathbf{y}}}
 \def\bh{\mathbf{h}}

\def\bA{{\mathbf{A}}}    
\def\bF{{\mathbf{F}}}  \def\bH{{\mathbf{H}}} \def\bI{{\mathbf{I}}} 
  \def\bM{{\mathbf{M}}} \def\bN{{\mathbf{N}}} 
   \def\bS{{\mathbf{S}}} 
\def\bU{{\mathbf{U}}} \def\bV{{\mathbf{V}}} \def\bW{{\mathbf{W}}}  \def\bY{{\mathbf{Y}}}

\def\argmin{\mathop{\mathrm{argmin}}}

\def\rank{\mathop{\mathrm{rank}}}
\def\span{\mathop{\mathrm{span}}}
\def\vec{\mathop{\mathrm{vec}}}

\def\diag{\mathop{\mathrm{diag}}}

\renewcommand{\baselinestretch}{1.0}

     \def\d4{\!\!\!\!}

\def\bSig{\mathbf{\Sigma}}

  \def\R{{\mathbb{R}}} \def\C{{\mathbb{C}}}   \def\E{{\mathbb{E}}}

  \def\la{\left|}     \def\ra{\right|}    \def\lA{\left\|}     \def\rA{\right\|}

  \def\-{\! - \!}  \def\+{\! + \!}  \def\={\! = \!}  \def\>{\! > \!}

\newtheorem{proposition}{Proposition}

\newtheorem{lemma}{Lemma}

\newcommand{\bef}{\begin{figure}}
\newcommand{\eef}{\end{figure}}
\newcommand{\beq}{\begin{eqnarray}}
\newcommand{\eeq}{\end{eqnarray}}

\newenvironment{proof}[1][Proof]{\begin{trivlist}
\item[\hskip \labelsep {\bfseries #1}]}{\end{trivlist}}

\newcommand{\qed}{\nobreak \ifvmode \relax \else
\ifdim\lastskip<1.5em \hskip-\lastskip \hskip1.5em plus0em
minus0.5em \fi \nobreak \vrule height0.5em width0.5em
depth0.25em\fi}

\begin{document}
\begin{spacing}{1.0}
%add the numbering of the article
\pagenumbering{arabic}

\title{Two-stage Method for Millimeter Wave Channel Estimation}
\author{Wei Zhang, Shu-Hung Leung, and Taejoon Kim
\thanks{
{W. Zhang and S. H. Leung are with the Department of Electronic Engineering, City University of Hong Kong, , Kowloon, Hong Kong, China (e-mail: wzhang237-c@my.cityu.edu.hk, eeeugshl@cityu.edu.hk).}
%{T. Kim is with the Department of Electrical Engineering and Computer Science, University of Kansas, KS 66045, USA (e-mail: taejoonkim@ku.edu).}
{T. Kim is with the Department of Electrical Engineering and Computer Science, University of Kansas, KS 66045, USA (e-mail: taejoonkim@ku.edu).}
}
}
\maketitle
\begin{abstract}
  The millimeter wave is a promising technique for the next generation of mobile communication. The large antenna array is able to provide sufficient precoding gain to overcome the high pathloss at millimeter wave band. However, the accurate channel state information is the key for the precoding design. Unfortunately, the channel use overhead and complexity are two major challenges when estimating the channel with high-dimensional array. In this paper, we propose a two-stage approach which reduces the channel use overhead and the computational complexity. Specifically, in the first stage, we estimate the column subspace of the channel matrix. Based on the estimated column subspace, we design the training sounders to acquire the remaining coefficient matrix of the column subspace. By dividing the estimation task into two stages, the training sounders for the second stages are only targeted for the column subspace, which will save the channel uses and the computational complexity as well.
\end{abstract}
\section{Introduction}
The millimeter wave band (30GHz $\sim$ 300GHz) is able to provide large bandwidth for the mobile communication, which make the millimeter wave communication a potential candidate for the 5G cellular networks.
Compared with the current 4G cellular networks, due to the high frequency of millimeter wave band, the signal will experience pathloss. Fortunately, the short wavelength enables the transmitter and receiver to be equipped with large antenna array, which will provide sufficient precoding gain to overcome the high pathloss.

Due to short wavelength, the millimeter wave channel experiences sparse property, which means the number of paths between transmitter and receiver is quite limited. Therefore, one possible method is to search the paths of the channel. Intuitively, one can try all the possible precoders and combiners for the transmitter and receiver respectively, then select the pair of precoder and combiner which provides the highest channel gain. The drawback of this exhaustive search method is that the extremely high searching overhead when the dimension of antenna array is large. In order to reduce the search overhead, the hierarchical codebook is proposed where the searching task is divided into several layers, and each layer will seek the required resolution of angles of arrival (AoAs) and angles of departure (AoDs). To further enhance the estimation accuracy, the generalized detection probability \cite{codebookDesign} is proposed as metric to guide the design of codebook.

Due to the fact that the millimeter wave channel experience sparse scattering, the sparse signal recovery method can be adopted to obtain the estimation of channel, which will save the number of required observations to estimate the channel. In \cite{ompEstimation}, the orthogonal matching pursuit (OMP) is utilized to recover the sparse vector associated with the position of AoAs and AoDs in the angle grids.
Due to the low-rank characteristic of channel matrix, \cite{ZhangSparse} proposed a sparse subspace decomposition method to recover the low-rank matrix.
In \cite{jointSparse}, the joint sparse and low-rank structure is utilized to model the millimeter wave channel, where the first step is to use the low-rank structure and the second step will utilize the sparse structure of the channel matrix. However, considering the degrees of freedom of the rank-$L$ matrix $\bH \in \C^{N_r \times N_t}$ is $d_{DoF}=L(N_r+N_t-L)$, the number of observations for the existing methods is quite larger than $d_{DoF}$. In other words, the channel use overhead is too large compared to the value of $d_{DoF}$. Moreover, when the sparse signal recover methods are adopted for channel estimation, these algorithms will experience high computational complexity.

In this paper, we propose a two-stage channel estimation method which require much lower channel use overhead and complexity as well. Specifically, in the first stage, we will sample the channel matrix to obtain the column subspace of the channel. Then, we can express the channel matrix by using the estimated column subspace.
Then in the second stage, we only need to estimate the coefficient matrix on the column subspace instead of the full space. By doing so, we will not waste our observations for the orthogonal subspace of the column subspace. Moreover, the coefficient matrix can be obtained by designing the corresponding receiver sounder effectively, which requires small computational complexity.

This paper is organized as follows, in section II, we introduce the signal and channel model in millimeter wave systems. Then,the column subspace sampling strategy is given by section III. By using the estimated column subspace, the estimation of  the remaining channel matrix is discussed in section IV. Finally, we evaluate the estimate accuracy of the proposed two-stage channel estimation method with the benchmarks.
\subsubsection*{Notations}
A bold lower case letter $\ba$ is a vector, a bold capital letter $\bA$ is a matrix.
$\bA^T,\bA^{\!-1}$, tr($\bA$), $| \bA  |$, $\| \bA \|_F$, $\| \bA  \|_*$, and $\|\ba\|_2$ are, respectively,    the transpose, inverse, trace, determinant, Frobenius norm, nuclear norm (i.e., the sum of the singular values of $\bA$) of $\bA$, and $l_2$-norm of $\ba$.
$[\bA]_{:,i}$, $[\bA]_{i,:}$, $[\bA]_{i,j}$, $[\ba]_i$ are, respectively, the $i$th column, $i$th row, $i$th row and $j$th column entry of $\bA$, and $i$th entry of vector $\ba$.
$\vec(\bA)$ stacks the columns of $\bA$ and form a long column vector.
$\diag(\bA)$ extracts the diagonal entries of $\bA$ to form a column vector.
$\bI_M \! \in \! \R^{M\times M}$ is the identity matrix.

\section{Signal Model}
In this section, we will introduce the millimeter wave system, channel model, and training signal structures.

We assume the transmitter and receiver to be equipped with $N_t$ and $N_r$ antennas, respectively.
Suppose the number of paths between transmitter and receiver be $L$, i.e., $L \ll min\{N_r, N_t\}$. The channel expression with the uniform linear array is given by
\beq
\bH = \sqrt{\frac{N_r N_t}{L}}\sum_{l=1}^{L} h_l \ba(\theta_{r,l}) \ba(\theta_{t,l}), \label{channel model}
\eeq
where $\ba(\theta_{r,l})$ and $\ba(\theta_{t,l})$ are the array response vectors of transmit and receive antenna array. The transmitter and receiver usually utilize the uniform linear array, where $\ba(\theta)$ is given by
\beq
\ba(\theta)=\frac{1}{\sqrt{N}}\left[ 1, e^{-j\frac{2\pi}{\lambda}d \sin \theta} , \ldots,e^{-j\frac{2\pi}{\lambda}d(N-1) \sin \theta}  \right]^T
\eeq
where $N$ is the number of antenna array at the transmitter or the receiver, $\lambda$ is the wavelength, $d=\frac{1}{2} \lambda$ is the antenna distance. Here, we assume $\theta_{r,l}$ and $\theta_{t,l}$ are both uniform distributed in $(0,2\pi)$, and the gain of paths has the following distribution $h_l \sim \cC \cN(0,1)$. Here, the channel model \eqref{channel model} in can be written as a matrix form,
\beq
\bH=\bA_r \diag(\bh) \bA_t^T, \label{matrix expression}
\eeq
where $\bA_r=[\ba(\theta_{r,1} ),\ldots,\ba(\theta_{r,L})] \in \C^{N_r \times L}$ and $\bA_t=[\ba(\theta_{t,1} ),\ldots,\ba(\theta_{t,L})]\in \C^{N_t \times L}$. The $i$th column of $\bH$ can be written as
\beq
\bH_i = h_1\ba(\theta_{r,1})e^{-j\pi(i-1)\sin\theta_{t,1}}\cdots h_L\ba(\theta_{r,L})e^{j\pi(i-1)\sin\theta_{t,L}}. \nonumber
\eeq

Suppose the sub-matrix $\bH_S = \bH \bS \in \C^{N_r \times m}$ which selects the first $m$ columns of $\bH$. In the following, we will show that when $m \ge L$, $\bH_S$ shares the same column subspace with $\bH$.

\begin{lemma}
When $m \ge L$, then $\bA_t^T \bS$ is a full row rank matrix. Moreover, the column subspace of $\bH_S=\bA_r diag(\bh) \bA_t^T \bS$ is equal to the one of $\bH$.
\end{lemma}
\begin{proof}
  From the expression \eqref{matrix expression}, we can know that $\bA_t$ and $\bA_r$ are full rank when the angle $\{\theta_{t,l}\}_{l=1}^L$ are distinct.
  Due to the fact that $\bH_S=\bA_r diag(\bh) \bA_t^T \bS$, in order to show the equivalence of the column subspace between $\bH$ and $\bH_S$, it is sufficient to show that $\bA_t^T \bS$ is full row rank. According to \cite{spatially}, when $m \ge L$, $\bA_t^T \bS$ will be a full row rank matrix. This concludes the proof.
\end{proof}

Suppose there are $N_{RF}$ RF chains at the transmitter and receiver side, the sampling observation at the $k$th channel use is given by
\beq
\by_k= \bW_k^*\bH \bff_k+\bW_k^* \bn_k,
\eeq
where $ \bW_k =  \bW_{A,k} \bW_{D,k} \in \C^{N_r \times N_{RF}}$.  $\bW_{A,k} \in \C^{N_r \times N_{RF}}$ and $\bW_{D,k}\in \C^{N_{RF} \times N_{RF}}$ are the receiver analog and digital sounder, respectively. The $\bff_k =  \bF_{A,k} \bF_{D,k} \bs_k \in \C^{N_t \times 1}$ is the transmitter sounder. Considering that $N_{RF} \ge 2$, the entries in $\bff_k$ can be arbitrary value at each channel use. Note that we have the power constraint for the sounding in each channel use$\lA \bff_k\rA_F^2 \le 1$, and noise $\bn_k \sim \cN(\boldmath{0},\sigma^2\bI)$. Thus, the signal to noise ratio is $1/\sigma^2$.
\begin{figure}[t]
\centering
\includegraphics[width=3 in]{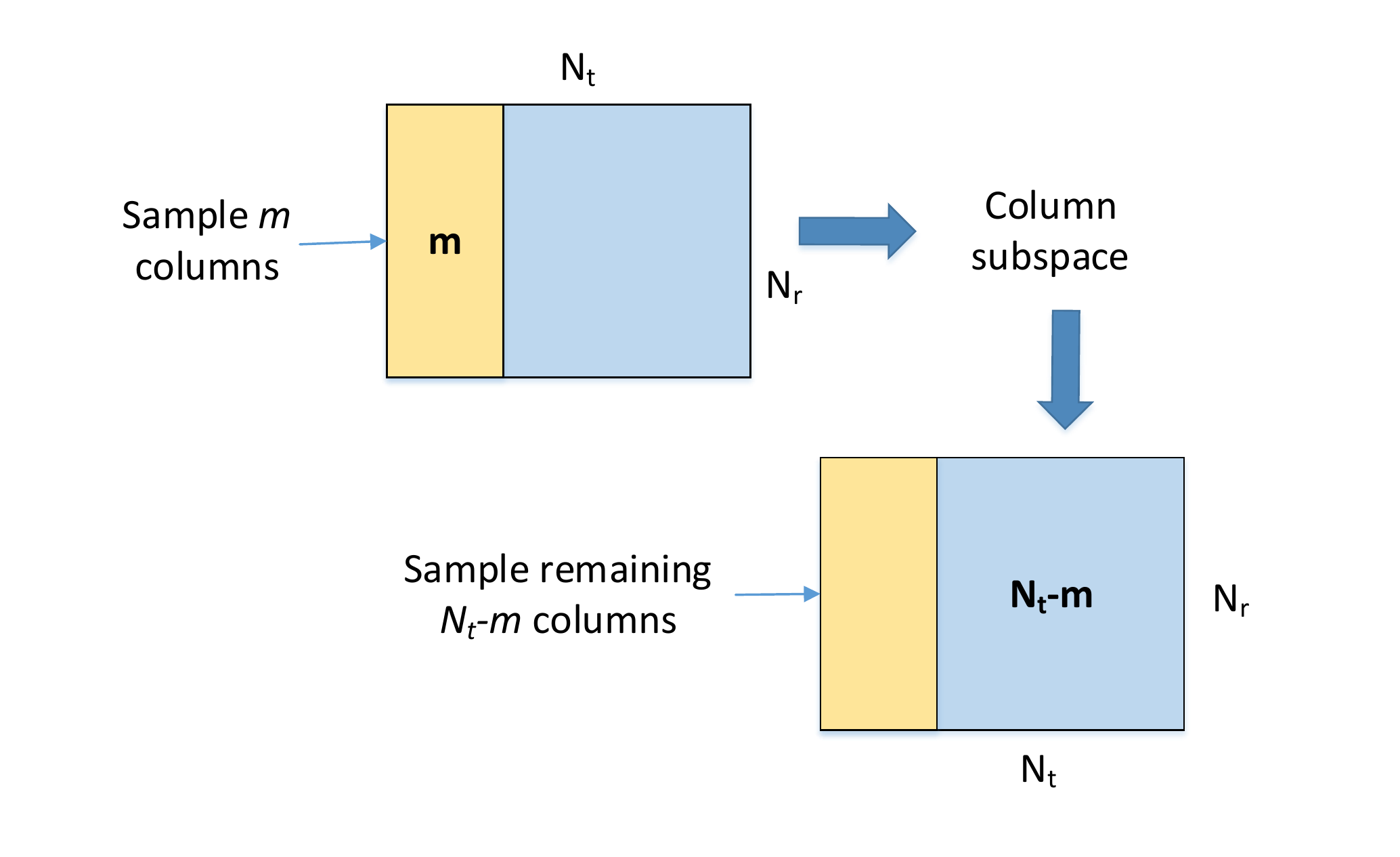}
\caption{The illustration of Algorithm} \label{framework}
\end{figure}

\section{Sample the column subspace}

In this section, we will show how to obtain the column subspace of the channel matrix.

Provided that the first $m$ ($m\ge L$) columns of $\bH$ has the same column subspace with $\bH$. To obtain the column subspace of $\bH$, we can just sample the first $m$ columns of $\bH$.
To illustrate how to get one single column of $\bH$ under the millimeter wave system, we take the first column of $\bH$ for example, i.e., $\bh_1$.
At the transmitter, we let $\bff_1=\be_1$, where $\be_1$ denotes the zero vector except the first entry is one. At the receiver side, we need $N_r/N_{RF}$  channel uses to form the full-rank matrix $\bM \in C^{N_r \times N_r }$,
\beq
\bM = [\bW_1, \bW_2,\cdots,\bW_{N_r/N}],
\eeq
where $\bW_i, i=1,2,\cdots,N_r/N_{RF}$ denotes the combiner at the receiver side.  Then we stack the observations of  $N_r/N_{RF}$  channel uses as $\bb_1 = [\by_1, \by_2,\ldots,\by_{N_r/N}]$,
\beq
 \left[
\begin{matrix}
  \by_1 \\
 \by_2 \\
  \vdots \\
  \by_{\frac{N_r}{N_{RF}}}
\end{matrix}
\right]
=
\left[
\begin{matrix}
  \bW_1^* \\
 \bW_2^* \\
  \vdots \\
  \bW_{\frac{N_r}{N_{RF}}}^*
\end{matrix}
\right]
\bH \be_1 +
\left[
\begin{matrix}
  \bW_1^* \\
 \bW_2^* \\
  \vdots \\
  \bW_{\frac{N_r}{N_{RF}}}^*
\end{matrix}
\right] \bn_1
=
\bM^*
\bh_1 +
\bM^*\bn_1 \label{H1 observation}
\eeq
\begin{proposition}
Suppose that we obtain the estimation of $\bh_1$ by solving least square problem as follows,
\beq
\min _{\bh_1}\lA \bb_1 - \bM\bh_1 \rA_F^2. \nonumber
\eeq
 If the $\bM$ in \eqref{H1 observation} is a full-rank matrix, then the estimation error of $\bh_1$ will only depend on $\bn_1$ in \eqref{H1 observation}. In other words, any full-rank matrix $\bM$ shares the same the estimation error.
\end{proposition}
\begin{proof}
From equation \eqref{H1 observation},
if we want to recover $\bh_1$ by solving the following problem
\beq
\min _{\bh_1}\lA \bb_1 - \bM\bh_1 \rA_F^2,
\eeq
then the solution is
\beq
\widehat{\bh}_1 &=& (\bM^*\bM)^{-1}\bM^*\bb_1 \nonumber\\
&=&(\bM^*\bM)^{-1}\bM^*(\bM\bh_1+\bM\bn_1)\nonumber \\
&=&\bh_1+\bn_1.
\eeq
The mean square error of $\bh_1$ is
\beq
\E \left( \lA \bh_1 -\widehat{\bh}_1  \rA_F^2\right) = \E \left( \lA \bn_1 \rA_F^2\right),
\eeq
which means that the MSE does not depend on $\bM$, which concludes the proof.
\end{proof}

Considering the analog constraint for $\bW_{A}$, such that $|[\bW_A]_{i,j}| = \frac{1}{\sqrt{N_r}}$, we let $\bM$ be the DFT matrix, which is obviously full-rank. After collecting all the observations, we have $\bY_S \in \C^{N_r \times m}$
\beq
\bY_S = \bM^* \bH_S+ \bM^*\bN, \label{expression1}
\eeq
where $\bH_S = [\bH]_{1:m}$ is the sub-matrix which selects the first $m$ columns of $\bH$, and $\bN \in \C^{N_r \times m}$ denotes the noise matrix with $[\bN]_{i,j} \sim \cC \cN(0,\sigma^2)$. Our next step is to obtain the column subspace of $\bH_S$ from the noisy observation $\bY_S$.
Since $\bM$ is full-rank, we rewrite the expression in \eqref{expression1} as
\beq
\tilde{\bY}_S = (\bM^*)^{-1}\bY_S =  \bH_S+ \bN.
\eeq
Due to the low-rank property of $\bH_S$, we estimate $\bH_S$ by solving the following standard principle analysis problem,
\beq
\widehat{\bH}_S = \argmin _{\bH_S} \lA \tilde{\bY}_S - \bH_S \rA_F^2  \nonumber \\
\text{subject to } \rank(\bH_S)=L. \label{pca}
\eeq
Then the column subspace of the channel matrix is given by the dominate $L$ left singular vectors of $\widehat{\bH}_S$, i.e., $col(\widehat{\bU})$. Note that the number of channel uses in the first step is $K_1 = (mN_r)/N_{RF}$.

Due to the noisy component, i.e., $\bN$ in \eqref{expression1}, we will analyze the subspace estimation accuracy of $col(\widehat{\bU})$ compared to the true column subspace $col(\bU)$ of $\bH_S$. According to \cite{cai2018}, the perturbed SVD will have the following bounds
\beq
\lA \bU \bU^* - \widehat{\bU} \widehat{\bU}^* \rA_2^2 \le \frac{CN_r\left(\sigma_L^2(\bH_S)\sigma^2+m \sigma^4\right)}{\sigma_L^4(\bH_S)} \wedge 1
 ,
\eeq
where $C$ is a constant, $\sigma_L(\bH_S)$ is the $L$th largest singular value of $\bH_S$, and $\sigma^2$ is the noise level of entries in $\bN$. We expand the right hand,
\beq
\!\!\!\!\frac{CN_r\left(\sigma_L^2(\bH_S)\sigma^2+m \sigma^4\right)}{\sigma_L^4(\bH_S)}\!\!\!\!\! & = &\!\!\!\!
 C\!\! \left(\! \frac{\sigma^2 N_r}{\sigma^2_L(\bH_S)} \! +\!  \frac{N_r m \sigma^4}{\sigma_L^4(\bH_S)} \!\! \right) \label{bound1}
\eeq
Since $\bH_S \in \C^{N_r \times m}$, when the value of $m$ increases, the value of $\sigma_L(\bH_S)$ will change accordingly. In the following, we will show the values of two parts on the right hand side of \eqref{bound1} will decrease with $m$.
\begin{lemma} \label{pro1}
Suppose $\bH_{S}$ selects first $m$ columns of $\bH$, and $\tilde{\bH}_{S} = [ \bH_{S}, \bh]$, which add one more column with respect to $\bH_{S}$. Then, we have the following
\beq
0 \le \sigma_L^2(\tilde{\bH}_{S}) - \sigma_L^2(\bH_{S}) \le |a_{L}|^2,
\eeq
where $\ba = \bU^* \bh$. Let the SVD of $\bH_S$ is given by $\bH_{S} = \bU \bSig \bV^*$, where $\bU \in \C^{N_r \times L}$, $\bV \in \C^{N_r \times L}$, and $\bSig \in \R^{L \times L}$ is in decreasing order.
\end{lemma}

\begin{proof}
Since we want to derive the relationship between the least singular value of $\bH_S$ and $\tilde{\bH}_{S}$, we can turn to the least eigenvalue of $\tilde{\bH}_{S} \tilde{\bH}_{S}^*$ and $\bH_{S} \bH_{S}^*$.
\beq
\tilde{\bH}_{S} \tilde{\bH}_{S}^* &=& \bU \bSig \bU^* + \bh \bh^* \nonumber\\
&=&\bU \bSig \bU^* + \bU \ba \ba^* \bU^*  \nonumber \\
&=&\bU (\bSig + \ba \ba^*)  \bU^*. \label{eigen1}
\eeq
Assume the eigenvalue decomposition of $(\bSig + \ba \ba^*)$ is $\tilde{\bU}_L \tilde{\bSig}_L \tilde{\bU}_L^*$. Thus, \eqref{eigen1} can be written as
\beq
\bU (\bSig + \ba \ba^*)  \bU^* = \bU \tilde{\bU}_L \tilde{\bSig}_L \tilde{\bU}_L^*  \bU^*.
\eeq
Then, the $L$th largest eigenvalue of $\tilde{\bH}_{S} \tilde{\bH}_{S}^*$ can be obtained from $\tilde{\bSig}_L$,
\beq
\sigma_L^2(\bH_{S}) = \lambda_L(\bH_{S} \bH_{S}^*) = \lambda_L(\bSig_L)\\
\sigma_L^2(\tilde{\bH}_{S}) = \lambda_L(\tilde{\bH}_{S} \tilde{\bH}_{S}^*) = \lambda_L(\bSig + \ba \ba^*) =  \lambda_L(\tilde{\bSig}_L).
\eeq
Let $\bv$ be the least eigenvector of $\bSig$, and $\tilde{\bv}$ is for $(\bSig + \ba \ba^*)$. We can have the following inequality
\beq
\tilde{\bv}^* (\bSig + \ba \ba^*)\tilde{\bv} \ge \tilde{\bv}^* \bSig \tilde{\bv} \ge \bv^* \bSig \bv.
\eeq
Therefore,
\beq
\sigma_L^2(\tilde{\bH}_{S}) - \sigma_L^2(\bH_{S}) \ge 0 \label{inequality1}.
\eeq
Meanwhile, due to the inequality $\bv^* (\bSig + \ba \ba^*) \bv \ge \tilde{\bv}^* (\bSig + \ba \ba^*)\tilde{\bv}$, combining the equation \eqref{inequality1} gives
\beq
\tilde{\bv}^* (\bSig + \ba \ba^*) \tilde{\bv}  - \bv^* \bSig \bv &\le&  \bv^* (\bSig + \ba \ba^*) \bv  - \bv^* \bSig \bv \nonumber \\
 &=&\bv^* (\bSig + \ba \ba^*) \bv  - \bv^* \bSig \bv \nonumber \\
  &=&\bv^*  \ba \ba^* \bv  \nonumber \\
 & =&  |a_L|^2,
\eeq
where $\ba = \bU^* \bh$, and $a_L$ is the $L$th element of $\ba$. This concludes the proof.
\end{proof}

The Lemma \ref{pro1} shows that the least singular value of $\bH_S$, i.e., $\sigma_L(\bH_S)$, will increase with $m$. Thus the value of first part in \eqref{bound1} will decrease with $m$.
For the second part, from the upper bound in Lemma \ref{pro1}, $\sigma^4_L(\bH_S)$ will increase approximately  2-order with $m$, i.e., $\mathcal{O}(m^2)$. Therefore, the value of second part will have ${N_r m \sigma^4}/{\sigma_L^4(\bH_S)}\sim \mathcal{O}(1/m)$, which means it also experiences decreasing with $m$.
To sum up, if we sample more columns in the first stage, we can acquire more accurate column subspace of the channel matrix. In the simulation part, we will validate this proposition.
\begin{algorithm} [t]
\caption{Two-stage millimeter wave channel estimation}
\label{alg1}
\begin{algorithmic} [1]
\STATE Input: channel dimension: $N_r$, $N_t$; number of RF chains: $N_{RF}$; channel paths: $L$.
\STATE Initialization: Generate the column selecting matrix $\bS \in\R^{ N_t \times m}$.
\STATE Column subspace learning:

(1).~Sample $i$th column of $\bH$ in $i_1,i_2,\cdots,i_{N_r/N_{RF}}$ channel use. The receiver sounder is
$\bM = [\bW_1, \bW_2,\cdots,\bW_{N_r/N_{RF}}]$,
the transmitter sounder is $\bff_i = \be_i, ~i=1,2,\cdots,m$.

(2). The column subspace is obtained by solving \eqref{pca}, which is given by the dominate $L$ left singular vector of $\widehat{\bH}_S$, i.e., $\widehat{\bU}$.

\STATE Sample the remaining matrix:

(1). Let the transmit sounder as $\be_i$, and design the receiver sounder according to \eqref{receiver sounder}.

(2). \!\!The observation is $\by_i = \widehat{\bW}^*\bH\bff_i +\widehat{\bW}^*\bn_i$. The estimation of $i$th column is given by $\widehat{\bh}_i = \widehat{\bW} \widehat{\bW}^* \bh_i, ~i=1,2,\cdots,N_t-m$.

(3). Stack the $N_t-m$ columns and denote it as $\widehat{\bH}_R$.
\STATE Output: Estimate result $\widehat{\bH} = [\widehat{\bH}_S, \widehat{\bH}_R]$.
\end{algorithmic}
\end{algorithm}

\section{Learn the remaining matrix}
In this section, we will show how to learn the coefficient matrix associated with the estimated column subspace $\widehat{\bH}$, i.e., $\widehat{\bU}$.

Since we have already sample the first $m$ columns of $\bH$ in the first step, thus we only need to sample the remaining $N_t - m$ columns. For the $i$th column, we also let the transmitter sounder $\bff_i=\be_i$ at the transmitter side. At the receiver side, since the estimated column subspace is $\widehat{\bU}$,
we design the analog sounder and digital sounder by solve the following
\beq
\left(\widehat{\bW}_A, \widehat{\bW}_D\right) = \min_{\bW_A,\bW_D} \lA \widehat{\bU}-\bW_A\bW_D \rA_F^2, \\
\text{subject to~~} \la[\bW_A]_{i,j}\ra=\frac{1}{\sqrt{N_r}} \label{receiver sounder}
\eeq
The problem above can be solved by using the OMP algorithm \cite{spatially}. We let the receiver sounder at the second step be $\widehat{\bW} = \widehat{\bW}_A \widehat{\bW}_D$.

First of all, we analyze how many channel uses are needed in the second step. In each channel use, we can sample one column of $\bH$ by letting the transmitter sounder and receiver sounder as $\bff_i$ and $\widehat{\bW}$, respectively. For each column, the number of channel uses is $1$. Therefore, in order to obtain the remaining $N_t-m$ columns, the required number of channel uses for the second stage is
\beq
K_2 = N_t - m.
\eeq

Now, for the $i$th column, the observation at the receiver side is in the following expression
\beq
\by_i = \widehat{\bW}^*\bH\bff_i +\widehat{\bW}^*\bn_i = \widehat{\bW}^*\bh_i + \widehat{\bW}^*\bn_i,
\eeq
where $\bn_i \sim \cC \cN(\boldmath{0}, \sigma^2\bI_{N_{RF}})$ is the noise vector. Similarly, the estimation of $\bh_i$ can be obtained by solving the following problem
\beq
\widehat{\bh}_i=\min_{\bh_i} \lA \by_i - \widehat{\bW}^*\bh_i \rA_F^2. \label{stage2}
\eeq
The problem \eqref{stage2} has more than one solution. It is because that $\widehat{\bh}_i$ and $\widehat{\bh}_i+\bd$ with $\bd \in \text{Null}(\widehat{\bW}^*)$ share the same objective value. Among all the solutions, the solution which has the smallest 2-norm is given by
\beq
\widehat{\bh}_i = \widehat{\bW} \by_i.
\eeq
After $(N_t-m)$ channel uses, the columns estimation $\bH_R = \widehat{\bH}_{m+1:N_t} \!\!= \!\! [\bh_{m+1},\bh_{m+2},\cdots,\bh_{N_t}]$ will be obtained. Combining the estimation result $\widehat{\bH}_S$ from the first stage, the estimation of the channel matrix is given by $\widehat{\bH}=[\widehat{\bH}_S, \widehat{\bH}_R]$. The whole algorithm is given in Algorithm \ref{alg1}.

In the following, we will analyze the complexity of the proposed two-stage method. We divide the complexity into two parts: channel use overhead and computational complexity.
For the computational complexity, the main complexity comes from the first stage, where we compute the column subspace through SVD. Specifically, the computational complexity for the first stage is $ \mathcal{O}(m^2N_r)$.
Note that for the total channel uses, after combining the necessary channel uses for the two stages, we can have
 \beq
 K =  \frac{m N_r}{N_{RF}} + (N_t-m). \label{channel uses}
 \eeq
Therefore, the number of channel uses is quite low. Specifically, when we let $m=L$, the number of channel uses will be ${L N_r}/{N_{RF}} + (N_t-L)$. Thus, if more RF chains are adopted, the required channel uses can be further reduced.

\section{Simulation Results}
In this section, we numerically evaluate the NMSE of the proposed two-stage channel estimation algorithm.
The numerical simulation setting is first discussed.
\begin{figure}[t]
\centering
\includegraphics[width=3 in]{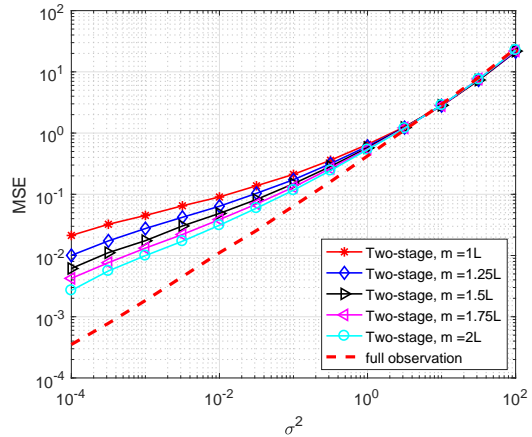}
\caption{Performance evaluation} \label{effect_of_m}
\end{figure}

(1) \emph{Channel model}:
We assume the prevalent physical channel representation  that models sparse millimeter wave MIMO channels \cite{RobertOver, brady}.
We assume the dimension of the channel is $N_r = 32$, and $N_t=128$. We further assume $L=4$ paths, and $N_{RF} =6$ RF chains.

(2) \emph{Performance Evaluation}:
The NMSE statistics across different noise levels and channel uses are evaluated.
Each curve is obtained after averaging over $200$ channel realizations.

First of all, we illustrate the effect of $m$ on the estimation accuracy in Fig. \ref{effect_of_m}. Here, according to \eqref{channel uses}, different values of $m$ mean different channel uses. With the increasing of $m$ in the first step, which means that more columns are selected in the first step to estimate the column subspace. As a result, more accurate estimation of the column subspace will be obtained. As we can see from Fig. \ref{effect_of_m}, when the value of $m$ increases, the MSE will experience a reduction to some extent. The red dot line denotes the scenario that we observe each entry of the channel matrix, and the estimation is given by solving the following problem
\beq
\min_{\bH} \lA \bH - \bY \rA_F^2 \text{~~subject to~~} \rank(\bH) = L,
\eeq
where $\bY = \bH + \bN$ is the noisy observation of the channel matrix, and the elements in $\bN$ are i.i.d. Gaussian. As it is shown, there is a gap between the two-stage method with the result of full observations. This is because that the error exists in the estimation of the column subspace in the first stage, which will further affect the estimation accuracy of remaining channel matrix in the second stage. For the scenario of full observation, there does not exist this effect.

Then, we simulate the accuracy of column subspace with different values of $m$ in Fig. \ref{subspace distance}. The subspace distance is defined as
$\lA \bU \bU^* - \widehat{\bU} \widehat{\bU}^* \rA_2^2$.
As we can see in Fig. \ref{subspace distance}, when $m=2L$, we can acquire an accurate column subspace from the first stage. Also, with the increasing of $m$, the subspace estimation accuracy will be improved.

\begin{figure}[t]
\centering
\includegraphics[width=3 in]{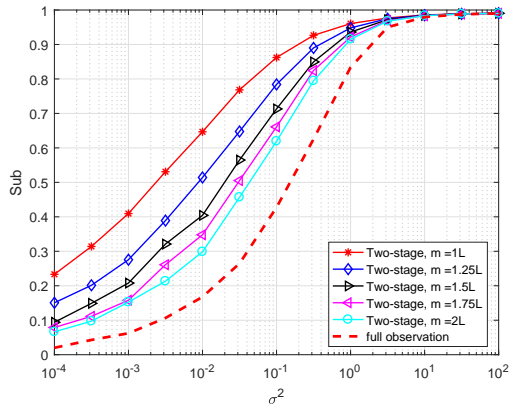}
\caption{Performance evaluation for subspace accuracy} \label{subspace distance}
\end{figure}

\section{Conclusion}
We divide the channel estimation task into two stages, where the first stage is to obtain the column subspace of the channel matrix, and the second stage will rely on the acquired subspace information to get the remaining channel matrix. As we analyzed, when more columns are sampled in the first stage, the column subspace information will be more accurate, which will bring a more accurate channel estimation result. Specifically, considering  the DoFs of the rank-$L$ matrix is $(N_r + N_t-L)L$, the number of channel uses is  only $K =  {m N_r}/{N} + (N_t-m)$. Therefore, the channel use overhead and computational complexity will benefit from the proposed method.
\bibliographystyle{IEEEtran}
\bibliography{IEEEabrv,Conference_mmWave_CS}
\clearpage

\end{spacing}
\end{document}